\documentclass[sigconf]{acmart}

\copyrightyear{2023}
\acmYear{2023}
\setcopyright{rightsretained}
\acmConference[CIKM '23]{Proceedings of the 32nd ACM International Conference on Information and Knowledge Management}{October 21--25, 2023}{Birmingham, United Kingdom}
\acmBooktitle{Proceedings of the 32nd ACM International Conference on Information and Knowledge Management (CIKM '23), October 21--25, 2023, Birmingham, United Kingdom}
\acmDOI{10.1145/3583780.3614712}
\acmISBN{979-8-4007-0124-5/23/10}

\usepackage{natbib}
\usepackage{graphicx}
\usepackage{blindtext}
\usepackage{hyperref}
\usepackage{booktabs}
\usepackage{bm}
\usepackage{bbm}
\usepackage{caption}
\usepackage{subcaption}
\usepackage{multirow}
\usepackage{xspace}
\usepackage{caption}
\usepackage{setspace}
\usepackage{caption}
\usepackage{balance}

\newcommand{\SigmoidCE}{{\mathrm{SigmoidCE}}}
\newcommand{\SoftmaxCE}{{\mathrm{SoftmaxCE}}}
\newcommand{\PairwiseLogistic}{{\mathrm{PairwiseLogistic}}}
\newcommand{\ListCE}{{\mathrm{ListCE}}}
\newcommand{\LogLoss}{{\mathrm{LogLoss}}}

\newcommand{\nsig}{\phantom{$^{\scriptscriptstyle\blacktriangle}$}\xspace}
\newcommand{\sig}{$^{\scriptscriptstyle\blacktriangle}$\xspace}
\newcommand{\sigworse}{$^{\scriptscriptstyle\blacktriangledown}$\xspace}

\DeclareMathOperator{\E}{\mathbbm{E}}
\DeclareMathOperator{\I}{\mathbbm{I}}

\newtheorem{definition}{Definition}

\newtheorem{proposition}{Proposition}

\begin{document}

\title{Regression Compatible Listwise Objectives for Calibrated Ranking with Binary Relevance}

\author{Aijun Bai}
\affiliation{%
  \institution{Google LLC}
\city{Mountain View}
 \country{USA}
  }
\email{aijunbai@google.com}

\author{Rolf Jagerman}
\affiliation{%
  \institution{Google LLC}
\city{Amsterdam}
 \country{Netherlands}
  }
\email{jagerman@google.com}

\author{Zhen Qin}
\affiliation{%
  \institution{Google LLC}
\city{Mountain View}
 \country{USA}
  }
\email{zhenqin@google.com}

\author{Le Yan}
\affiliation{%
  \institution{Google LLC}
\city{Mountain View}
 \country{USA}
  }
\email{lyyanle@google.com}

\author{Pratyush Kar}
\affiliation{%
  \institution{Google LLC}
\city{Paris}
 \country{France}
  }
\email{pratk@google.com}

\author{Bing-Rong Lin}
\affiliation{%
  \institution{Google LLC}
\city{Mountain View}
 \country{USA}
  }
\email{bingrong@google.com}

\author{Xuanhui Wang}
\affiliation{%
  \institution{Google LLC}
\city{Mountain View}
 \country{USA}
  }
\email{xuanhui@google.com}

\author{Michael Bendersky}
\affiliation{%
  \institution{Google LLC}
\city{Mountain View}
 \country{USA}
  }
\email{bemike@google.com}

\author{Marc Najork}
\affiliation{%
  \institution{Google LLC}
\city{Mountain View}
 \country{USA}
  }
\email{najork@google.com}

\renewcommand{\shortauthors}{Bai et al.}

\begin{abstract}
As Learning-to-Rank (LTR) approaches primarily seek to improve ranking quality, their output scores are not scale-calibrated by design. This fundamentally limits LTR usage in score-sensitive applications. Though a simple multi-objective approach that combines a regression and a ranking objective can effectively learn scale-calibrated scores, we argue that the two objectives are not necessarily compatible, which makes the trade-off less ideal for either of them. In this paper, we propose a practical regression compatible ranking (RCR) approach that achieves a better trade-off, where the two ranking and regression components are proved to be mutually aligned. Although the same idea applies to ranking with both binary and graded relevance, we mainly focus on binary labels in this paper. We evaluate the proposed approach on several public LTR benchmarks and show that it consistently achieves either best or competitive result in terms of both regression and ranking metrics, and significantly improves the Pareto frontiers in the context of multi-objective optimization. Furthermore, we evaluated the proposed approach on YouTube Search and found that it not only improved the ranking quality of the production pCTR model, but also brought gains to the click prediction accuracy. The proposed approach has been successfully deployed in the YouTube production system.
\end{abstract}



\maketitle

\section{Introduction}
Learning-to-Rank (LTR) aims to construct a ranker from training data such that it can rank unseen objects correctly. It is therefore required that a ranker performs well on ranking metrics, such as {\em Normalized Discounted Cumulative Gain} (NDCG). It is usually the case that a ranking-centric pairwise or listwise approach, such as RankNet \cite{burges2010ranknet} or ListNet \cite{xia2008listwise}, achieves better ranking quality than a regression approach that adopts a pointwise formulation.

On the other hand, modern systems in these applications have multiple stages and downstream stages consume the predictions from previous ones. It is often desired that the ranking scores are well calibrated and the distribution remains stable.  Take the application of online advertising as an example. In online advertising, the pCTR (predicted Click-Through Rate) model is required to be well calibrated because it affects the downstream auction and pricing models~\cite{yan2022scale,ChengLi:KDD15,Calibrate:ADKDD17}, though the final ranking of ads is the one that matters most for the performance. This suggests that we want the ranker to perform well not only on ranking metrics, but also on regression metrics in terms of calibrating ranker output scores to some external scale. Popular regression metrics include Mean Squared Error (MSE) for graded relevance labels and the logistic loss (LogLoss) for binary relevance labels.

Unsurprisingly, capable ranking approaches perform poorly on regression metrics, due to the fact that their loss functions are invariant to rank-preserving score transformations, and tend to learn scores that are not scale-calibrated to regression targets. Furthermore, these approaches suffer from training instability in the sense that the learned scores may diverge indefinitely in continuous training or re-training \cite{yan2022scale}. These factors strongly limit their usage in score-sensitive applications. As a result, practitioners have no choice but to fall back on regression-only approaches even if they are suboptimal in terms of user-facing ranking metrics.

It has been shown that a standard multi-objective approach effectively learns scale-calibrated scores for ranking \cite{sculley2010combined,ChengLi:KDD15,yan2022scale,yan2023learning}. However, we argue that in this standard multi-objective setting, the regression and ranking objectives are inherently conflicting, and thus the best trade-off might not be ideal for either of them. In this paper, we propose a practical regression compatible ranking (RCR) approach where the two ranking and regression components are proved to be mutually aligned. Although the same idea applies to ranking with both binary and graded relevance, we mainly focus on binary labels in this paper. Empirically, we conduct our experiments on several public LTR datasets, and show that the proposed approach achieves either the best or competitive result in terms of both regression and ranking metrics, and significantly improves the Pareto frontiers in the context of multi-objective optimization. Furthermore, we evaluated the proposed approach on YouTube Search and found it not only improved the ranking capability of the production pCTR model but also brought gains to the click prediction accuracy. The proposed approach has been fully deployed in the YouTube production system.

\section{Related Work} \label{sec:relatedwork}

There is a long history of the study of Learning-to-Rank (LTR) \cite{liu2009learning}. The general set up is that a scoring function is trained to score and sort a list of objects given a context. The accuracy is evaluated based on ranking metrics that only care about the order of the objects, but not the scale of the scores. Existing works include designing more effective loss function \cite{cao2007learning, zhu2020listwise, jagerman2022optimizing,xia2008listwise}, learning from biased interaction data \cite{wang2016learning,joachims2017unbiased,jagerman2019model,revisit}, and different underlying models, from support vector machines \cite{Joachims:2002}, to gradient boosted decision trees \cite{ke2017lightgbm, wang2018lambdaloss}, to neural networks \cite{pang2020setrank, pasumarthi2019tf, lambdarank, matching, ranknet2005icml, dasalc}. However, almost none of the existing works have studied the calibration issue of the ranking model outputs, which limits their applicability in many applications where a calibrated output is necessary.

To the best of our knowledge, few works have studied the ranking calibration problem. Similar to classification problems, post-processing methods can be used for calibrating ranking model outputs. For example, \citeauthor{CTR:LTR:ADKDD13} \cite{CTR:LTR:ADKDD13} used the pairwise squared hinge loss to train an LTR model for ads ranking, and then used Platt-scaling \cite{PlattScaling:1999} to convert the ranking scores into probabilities. Recently, \citeauthor{Calibrate:ADKDD17} \cite{Calibrate:ADKDD17} compared different post-processing methods to calibrate the outputs of an ordinal regression model, including Platt-scaling and isotonic regression. Our proposed method does not rely on a post-processing step.

Another class of approaches is based on multi-objective setting where ranking loss is calibrated by a regression loss during the training time, without an additional post-processing step. \citeauthor{sculley2010combined} \cite{sculley2010combined} is an early work that combines regression and ranking. It has been used in concrete application \cite{yan2022scale,ChengLi:KDD15}. In particular, \citeauthor{yan2022scale} \cite{yan2022scale} used the multi-objective formulation in deep models to prevent models from diverging during training and achieve output calibration at the same time. The shortcoming of such an approach is that ranking accuracy can be traded for calibration because the two objectives are not designed to be compatible. Our proposed method does not sacrifice ranking accuracy to achieve calibration.

\section{Background} \label{sec:background}

Learning-to-Rank (LTR) concerns the problem of learning a model to rank a list of objects given a context. Throughout the paper, we use {\em query} to represent the context and {\em documents} to represent the objects. In the so-called {\em score-and-sort} setting, a ranker is learned to score each document, and a ranked list is formed by sorting documents according to the scores.

More formally, let $q \in Q$ be a query and $x \in \mathcal{X}$ be a document, a score function is defined as $s(q, x ; \bm{\theta}): Q \times \mathcal{X} \rightarrow \mathbb{R}$, where $Q$ is the query space, $\mathcal{X}$ is the document space, and $\bm{\theta}$ is the parameters of the score function $s$. A typical LTR dataset $D$ consists of examples represented as tuples $(q, x, y) \in D$ where $q$, $x$ and $y$ are query, document and label respectively. Let $\mathbf{q} = \{ q | (q, x, y) \in D \}$ be the query set induced by $D$. Let $\mathcal{L}_\mathit{query}(\bm{\theta}; q)$ be the loss function associated with a single query $q \in Q$. Depending on how $\mathcal{L}_\mathit{query}$ is defined, LTR techniques can be roughly divided into three categories: pointwise, pairwise and listwise.

In the pointwise approach, the query loss $\mathcal{L}_\mathit{query}$ is represented as sum of losses over documents sharing the same query. For example, in logistic-regression ranking (i.e. ranking with binary relevance labels), the Sigmoid Cross Entropy loss per document (denoted by SigmoidCE) is defined as:
\begin{equation} \label{eq:sigmoidce}
    \SigmoidCE(s, y) = -y \log \sigma(s) - (1 - y) \log(1 - \sigma(s)),
\end{equation}
where $s = s(q, x ; \bm{\theta})$ is the predicted score of query-document pair $(q, x)$ and $\sigma(s) = (1 + \exp(-s))^{-1}$ is the sigmoid function. SigmoidCE is shown to be scale-calibrated \cite{yan2022scale} in the sense that it achieves global minima when
$
    \sigma(s) \to \E[y | q, x]
$.

In the pairwise approach, the query loss $\mathcal{L}_\mathit{query}$ is represented as sum of losses over all document-document pairs sharing the same query. The fundamental RankNet approach uses a pairwise Logistic loss (denoted by PairwiseLogistic) \cite{burges2010ranknet}:
\begin{equation}
    \PairwiseLogistic(s_1, s_2, y_1, y_2) = -\I(y_2 > y_1) \log \sigma(s_2 - s_1),
\end{equation}
where $s_1$ and $s_2 $ are the predicted scores for documents $x_1$ and $x_2$, $\I$ is the indicator function, and $\sigma$ is the sigmoid function. PairwiseLogistic achieves global minima when
$
    \sigma(s_2 - s_1) \to \E[\I(y_2 > y_1) | q, x_1, x_2]
$,
which indicates that the loss function mainly considers the pairwise score differences, which is also known as the \textit{translation-invariant} property~\cite{yan2022scale}.

In the listwise approach, the query loss $\mathcal{L}_\mathit{query}$ is attributed to the whole list of documents sharing the same query. The popular ListNet approach uses the Softmax based Cross Entropy loss (denoted by SoftmaxCE) to represent the listwise loss as \cite{xia2008listwise}:
\begin{equation}
        \SoftmaxCE(s_{1:N}, y_{1:N}) = - {1 \over C} \sum_{i=1}^{N} y_i \log {\exp(s_i) \over \sum_{j=1}^{N} \exp(s_j)},
\end{equation}
where $N$ is the list size, $s_i$ is the predicted score, and $C = \sum_{j=1}^{N}y_j$. The global minima will be achieved at \cite{xia2008listwise}:
\begin{equation}
    {\exp(s_i) \over \sum_{j=1}^{N} \exp(s_j) } \to { \E[y_i | q, x_i] \over \sum_{j=1}^{N} \E[y_j| q, x_j] }.
\end{equation}
Similar to PairwiseLogistic, the SoftmaxCE loss is translation-invariant, and could give scores that are arbitrarily worse with respect to regression metrics.

\section{Regression Compatible Ranking} \label{sec:main}

In this section, we first give the motivation, then formally propose the approach to regression compatible ranking (RCR).

\subsection{Motivation} \label{subsec:motivation}

It has been shown in the literature that a standard multi-objective approach effectively learns scale-calibrated scores for ranking \cite{sculley2010combined,ChengLi:KDD15,yan2022scale}. Taking logistic-regression ranking as an example, \citeauthor{yan2022scale} define the multi-objective loss as a weighted sum of SigmoidCE and SoftmaxCE losses:
\begin{multline}
    \mathcal{L}_\mathit{query}^\mathit{MultiObj}(\bm{\theta}; q) = (1 - \alpha) \cdot \sum_{i=1}^{N} \SigmoidCE(s_i, y_i) \\
    + \alpha \cdot \SoftmaxCE(s_{1:N}, y_{1:N}),
\end{multline}
where $\alpha \in [0, 1]$ is the trade-off weight. For simplicity, we refer to this method as $\SigmoidCE + \SoftmaxCE$. It can be seen that $\SigmoidCE + \SoftmaxCE$ is no longer translation-invariant, and has been shown effective for calibrated ranking. Let's take a deeper look at what scores are  learned following this simple multi-objective formalization.

Given query $q$, let $P_i = \E[y_i | q, x_i]$ be the ground truth click probability further conditioned on document $x_i$. Recall that, SigmoidCE achieves global minima when $\sigma(s_i) \to P_i$, which means we have the following pointwise learning objective for SigmoidCE:
\begin{equation} \label{eq:sigmoid_target}
    s_i \to \log P_i  - \log(1 - P_i).
\end{equation}

On the other hand, SoftmaxCE achieves global minima when
\begin{equation}
    {\exp(s_i) \over \sum_{j=1}^{N} \exp(s_j) } \to { P_{i} \over \sum_{j=1}^{N} P_{j} },
\end{equation}
or equivalently:
\begin{equation} \label{eq:softmax_target}
    s_i \to \log P_i - \log \sum_{j=1}^{N} P_{j} + \log \sum_{j=1}^{N} \exp(s_j),
\end{equation}
where the $\log$-$\sum$-$\exp$ term is an unknown constant and has no effects on the value or gradients of the final SoftmaxCE loss.

In the context of stochastic gradient descent, Equations \eqref{eq:sigmoid_target} and \eqref{eq:softmax_target} indicate that the gradients generated from the SigmoidCE and SoftmaxCE components are \textit{pushing the scores to significantly different targets}. This reveals the fact that the two losses in a standard multi-objective setting are inherently conflicting and will fail to find a solution ideal for both. How can we resolve this conflict?

Noticing that since $\sigma(s_i)$ is pointwisely approaching to $P_i$, if we replace the ground truth probabilities $P_i$ on the right side of Equation \eqref{eq:softmax_target} with the empirical approximations $\sigma(s_i)$ and drop the constant term, we are constructing some virtual logits:
\begin{equation} \label{eq:constructed_logits}
  s'_i \gets \log \sigma(s_i) - \log \sum_{j=1}^N \sigma(s_j).
\end{equation}

If we further apply SoftmaxCE loss on the new logits $s'_i$, we are establishing the following novel listwise learning objective:
\begin{equation}
  {\exp(s'_i) \over \sum_{j=1}^{N} \exp(s'_j) } \to {P_i \over \sum_{j=1}^N P_j},
\end{equation}
which is equivalent to
\begin{equation}  \label{eq:compound}
  {\sigma(s_i) \over \sum_{j=1}^N \sigma(s_j) } \to {P_i \over \sum_{j=1}^N P_j}.
\end{equation}

It is easy to see that Equation \eqref{eq:sigmoid_target} implies Equation \eqref{eq:compound} automatically, which means, as pointwise regression and listwise ranking objectives, they are well aligned in the sense that they achieve global minima simultaneously.

\subsection{The Main Approach}

Inspired by the above motivating example, we firstly define a novel Listwise Cross Entropy loss (ListCE) as follows.
\begin{definition}
Let $N$ be the list size, $s_{1:N}$ be the predicted scores, and $y_{1:N}$ be the labels. Let $T(s): \mathbb{R} \rightarrow \mathbb{R}^+$ be a non-decreasing transformation on scores. The Listwise Cross Entropy loss with transformation $T$ is defined as:
\begin{equation} \label{eq:listce}
\ListCE(T, s_{1:N}, y_{1:N}) = - {1 \over C} \sum_{i=1}^{N} y_i \log {T(s_i) \over \sum_{j=1}^{N} T(s_j)},
\end{equation}
where $C = \sum_{j=1}^{N}y_j$ is a normalizing factor.
\end{definition}

For the scope of this paper, we interchangeably use ListCE with transformation $T$, $\ListCE(T)$, or even ListCE when there is no ambiguity. We immediately have the following propositions.
\begin{proposition}$\ListCE(\exp)$ reduces to $\SoftmaxCE$.
\end{proposition}
\begin{proposition} \label{proposition:list_ce}
$\ListCE(T)$ achieves global minima when
\begin{equation}
    {T(s_i) \over \sum_{j=1}^{N} T(s_j) } \to { \E[y_i | q, x_i] \over \sum_{j=1}^{N} \E[y_j| q, x_j] }.
\end{equation}
\end{proposition}

\begin{proof}
Let $\overline{y} = \E[y | q, x]$ be the expected label of query-document pair $(q, x)$. Applying the ListCE loss on $(x, y) \in D$ is then equivalent to applying it on $(x, \overline{y})$ in expectation. Given transformation $T$, and predicted scores $s_{1:N}$, with $p_i = T(s_i) / \sum_{j=1}^{N} T(s_j)$, we have:
\begin{equation} \label{eq:listce_simplified}
\ListCE(T, s_{1:N}, \overline{y}_{1:N}) =  {1 \over \sum_{j=1}^{N} \overline{y}_j} \sum_{i=1}^{N} \overline{y}_i \log p_i,
\end{equation}
subject to
$\sum_{i=1}^{N} p_i = 1$.

Let's construct the following Lagrangian formalization:
\begin{equation} \label{eq:lagrangian}
    \mathcal{L}(p_{1:N}, \lambda) = {1 \over \sum_{j=1}^{N} \overline{y}_j} \sum_{i=1}^{N} \overline{y}_i \log p_i + \lambda (\sum_{i=1}^{N} p_i  1).
\end{equation}

Finding the extremum value of Equation \eqref{eq:listce_simplified} is then equivalent to finding the stationary points of Equation \eqref{eq:lagrangian}, which requires:
\begin{equation} \label{eq:partial1}
    {\partial{\mathcal{L}(p_{1:N}, \lambda)} \over \partial{p_i} } = {\overline{y}_i \over p_i \sum_{j=1}^{N} \overline{y}_j } + \lambda = 0,
\end{equation}
and
\begin{equation} \label{eq:partial2}
    {\partial{\mathcal{L}(p_{1:N}, \lambda)} \over \partial{\lambda} } = \sum_{i=1}^{N} p_i  1 = 0.
\end{equation}

Note that Equations \eqref{eq:partial1} and \eqref{eq:partial2} give us a system of $N+1$ equations on $N+1$ unknowns. It is easy to see that the unique solution is
\begin{equation}
p_i = {\overline{y}_i \over \sum_{j=1}^{N} \overline{y}_j},
\end{equation}
and $\lambda = 1$.

This indicates the unique global extremum at
\begin{equation}
    {T(s_i) \over \sum_{j=1}^{N} T(s_j) } \to  {\E[y_i | q, x_i] \over \sum_{j=1}^{N} \E[y_j | q, x_j]}.
\end{equation}

It is easy to verify that this unique global extremum attributes to the global minima which concludes the proof.
\end{proof}

In logistic-regression ranking, all labels are binarized or within the range of [0, 1]. A natural pointwise objective is the SigmoidCE loss. With SigmoidCE as the pointwise component, it is then required to use the sigmoid function as the transformation such that they can be optimized simultaneously without conflict.
\begin{definition}
The Regression Compatible Ranking (RCR) loss for a single query in a logistic-regression ranking task (i.e. ranking with binary relevance labels) is defined as:
\begin{multline} \label{eq:rcr_sigmoid}
    \mathcal{L}_\mathit{query}^\mathit{Compatible}(\bm{\theta}; q) = (1 - \alpha) \cdot \sum_{i=1}^{N} \SigmoidCE(s_i, y_i) \\
    + \alpha \cdot \ListCE(\sigma, s_{1:N}, y_{1:N}),
\end{multline}
where $\sigma$ is the sigmoid function.
\end{definition}

For simplicity, we refer to this method as $\SigmoidCE + \ListCE(\sigma)$. We have the following proposition:
\begin{proposition} \label{proposition:rcr_sigmoid}
$\SigmoidCE + \ListCE(\sigma)$ achieves global minima when $\sigma(s_i) \to \E[y_i | q, x_i]$.
\end{proposition}

\begin{proof}
The $\SigmoidCE$ component achieves global minima when $\sigma(s_i) \to \E[y_i | q, x_i]$ which implies
\begin{equation}
    {\sigma(s_i) \over \sum_{j=1}^{N} \sigma(s_j) } \to  {\E[y_i | q, x_i] \over \sum_{j=1}^{N} \E[y_j | q, x_j]},
\end{equation}
which minimizes $\ListCE(\sigma)$ at its global minima.
\end{proof}

\section{Experiments on Public Datasets} \label{sec:experiments}

To validate the proposed approach, we conduct our experiments on several public LTR datasets in this section.

\begin{table*}[ht]
\centering
\caption{Comparisons on logistic-regression ranking tasks. Model selection is done on validation sets with test set results reported. Numbers with bold font indicate the best result. \sig and \sigworse indicate statistical significance with p-value=0.05 of better and worse results than the {SoftmaxCE-Platt} baseline.}
\label{tab:binarized}
\begin{tabular}{l||r|r|r||r|r|r||r|r|r}
\hline
Datasets & \multicolumn{3}{c||}{Web30K} & \multicolumn{3}{c||}{Yahoo} & \multicolumn{3}{c}{Istella} \\ \hline
Metrics & NDCG@10 & LogLoss & ECE & NDCG@10 & LogLoss & ECE & NDCG@10 & LogLoss & ECE \\ \hline \hline
SigmoidCE & 0.4626\sig & \textbf{0.5996}\sig & \textbf{0.1216}\sig & 0.6852\sigworse & 0.4296\sig & 0.1807\sig & 0.6560\sigworse & \textbf{0.0612}\sig & 0.0275\sigworse \\ \hline
ListCE($\sigma$) & 0.4528\sigworse & 1.1675\sigworse & 0.1657\sigworse & 0.6954\sigworse & 0.7652\sigworse & 0.1475\sig & 0.6862\sig & 0.0643\sigworse & 0.0248\nsig \\
SoftmaxCE & 0.4578\nsig & 23.7719\sigworse & 0.5503\sigworse & 0.6993\nsig & 15.9964\sigworse & 0.2472\sigworse & 0.6839\nsig & 60.5913\sigworse & 0.9556\sigworse \\
{SoftmaxCE-Platt} & {0.4578}\nsig & {0.6103}\nsig & {0.1333}\nsig & {0.6993}\nsig & {0.5036}\nsig & {0.1962}\nsig & {0.6839}\nsig & {0.0628}\nsig & {0.0246}\nsig \\ \hline
SigmoidCE + SoftmaxCE & 0.4665\sig & 0.6239\sigworse & 0.1509\sigworse & 0.7008\sig & 0.4626\sig & 0.1852\sig & 0.6861\sig & 0.0643\sigworse & 0.0271\sigworse \\
\textbf{SigmoidCE + ListCE($\sigma$)} & \textbf{0.4680}\sig & 0.6031\sig & 0.1275\sig & \textbf{0.7050}\sig & \textbf{0.4187}\sig & \textbf{0.1550}\sig & \textbf{0.6900}\sig & 0.0634\nsig & \textbf{0.0242}\nsig \\ \hline
\end{tabular}
\end{table*}

\subsection{Experiment Setup}
\subsubsection{Datasets}

We extensively compare our methods with baselines on three datasets: Web30K \cite{qin2013introducing}, Yahoo \cite{chapelle2011yahoo}, and Istella \cite{dato2016fast}. These datasets have graded relevance labels. To study logistic-regression ranking, we simply binarize them by treating non-zero labels as 1s.

\textbf{Web30K} is a public dataset where the 31531 queries are split into training, validation, and test partitions with 18919, 6306, and 6306 queries respectively. There are on average about 119 candidate documents associated with each query. Each document is represented by 136 numerical features and graded with a 5-level relevance label. The percentages of documents with relevance label equal to 0, 1, 2, 3, and 4 are about 51.4\%, 32.5\%, 13.4\%, 1.9\%, and 0.8\%. When being binarized, the percentages for 0 and 1 are 51.4\% and 48.6\%.

\textbf{Yahoo} LTR challenge dataset consists of 29921 queries, with 19944, 2994 and 6983 queries for training, validation, and test respectively. There are 700 numerical features extracted for each query-document pair. The average number of documents per query is 24, but some queries have more than 100 documents. The labels are numerically graded. The distribution over 0, 1, 2, 3, and 4 is 21.9\%, 50.2\%, 22.3\%, 3.9\%, and 1.7\%. In binarized form, the distribution over 0 and 1 is 21.9\% and 78.1\%.

\textbf{Istella} LETOR dataset is composed of 33018 queries, with 20901, 2318, and 9799 queries respectively in training, validation, and test partitions. The candidate list to each query is with on average 316 documents, and each document is represented by 220 numerical features. The graded relevance labels also vary from 0 to 4 but with a more skewed distribution: 96.3\% for 0s, 0.8\% for 1s, 1.3\% for 2s, 0.9\% for 3s, and 0.7\% for 4s. With binarization, this distribution becomes 96.3\% for 0s and 3.7\% for 1s.

\subsubsection{Metrics}
In our experiments, we are interested in both regression and ranking performance. For ranking performance, we adopt the popular NDCG@10 \cite{jarvelin2002cumulated} as the main metric. More formally, given a list of labels $y_{1:N}$ and a list of output scores $s_{1:N}$, $\NDCG@k$ is defined as:
\begin{equation}
    \NDCG@k(s_{1:N}, y_{1:N}) = {\DCG@k(s_{1:N}, y_{1:N}) \over \DCG@k(y_{1:N}, y_{1:N})},
\end{equation}
where $\DCG@k$ is the so-called \textit{Discounted Cumulative Gain} up to position $k$ metric defined as:
\begin{equation}
    \DCG@k(s_{1:N}, y_{1:N}) = \sum_{i=1}^{N} \I(\pi(s_i) \leq k) { 2^{y_i} - 1 \over \log_2(\pi(s_i) + 1) },
\end{equation}
where $\pi(s_i)$ is the 1-based rank of score $s_i$ in the descendingly sorted list of $s_{i:N}$.

For regression performance, we mainly look at the LogLoss metric, which is defined as:
\begin{equation}
    \LogLoss(\hat{y}_{1:N}, y_{1:N}) = -{1 \over N} \sum_{i=1}^{N} y_i \log \hat{y}_i + (1 - y_i) \log (1 - \hat{y}_i),
\end{equation}
and
where $N$ is the total data size, $y_i$ is the label and $\hat{y}_i$ is the predicted label. Note that for LogLoss, $\hat{y}_i = \sigma(s_i)$ is the predicted probability after sigmoid transformation.

In addition, we consider the Expected Calibration Error (ECE) \cite{naeini2015obtaining,guo2017calibration} as a universal metric for calibration. This metric is commonly used in uncertainty calibration. Following \cite{yan2022scale}, we divide ranking documents in each query into $M$ bins after we sort them by the model predictions, and compute the ECE by,
\begin{equation}
    ECE = {1 \over |Q|} \sum_{q \in Q} \sum_{m=1}^{M} {|B_m| \over |D_q|} \left | {1 \over |B_m|} \sum_{i \in B_m} y_i - {1 \over |B_m|} \sum_{i \in B_m}  \hat{y}_i \right|.
\end{equation}
In this work, we use $M$ = 10 bins with each bin containing approximately the same number of documents with successive predictions.

\subsubsection{Methods}
We mainly compare the proposed $\SigmoidCE + \ListCE(\sigma)$ method with SigmoidCE and SigmoidCE + SoftmaxCE. Additionally, we include ListCE($\sigma$) in the comparison. We also compare with SoftmaxCE and SoftmaxCE-Platt, where the SoftmaxCE-Platt method applies the de facto Platt-scaling after a model that has been trained with SoftmaxCE.

We conduct our experiments using the TF-Ranking library \cite{pasumarthi2019tf}. In all experiments, we fix the ranker architecture to be a 3-layer Dense Neural Network (DNN) whose hidden layer dimensions are 1024, 512 and 256. The fraction of neuron units dropped out in training is set to be 0.5. We run the experiments on GPUs and use 128 as the training batch size. We use Adam \cite{kingma2014adam} as the optimizer, perform an extensive grid search of learning rates (LRs) and $\alpha$ over [0.01, 0.001] $\times$ [0.001, 0.005, 0.01, 0.05, 0.1, 0.25, 0.5, 0.75, 0.9, 0.95, 0.99, 0.995, 0.999] for 100 epochs. For each method, model selection is done on validation sets with test sets results reported. For pointwise regression SigmoidCE, we select a model by its regression performance on LogLoss, as this is what they are optimized for, and we are interested to see what the best regression performance they can achieve; for listwise, Multi-Objective and RCR methods, we select a model by its ranking performance on NDCG@10.

To further study the behavior of the approaches in the context of multi-objective optimization, we evaluate all models on the test data, and plot the Pareto frontier for each model in comparison.

\subsection{Experimental Results}

\subsubsection{Main Comparisons}

\begin{figure*}[ht]
     \centering
     \begin{subfigure}[b]{0.33\textwidth}
         \centering
         \includegraphics[width=\textwidth]{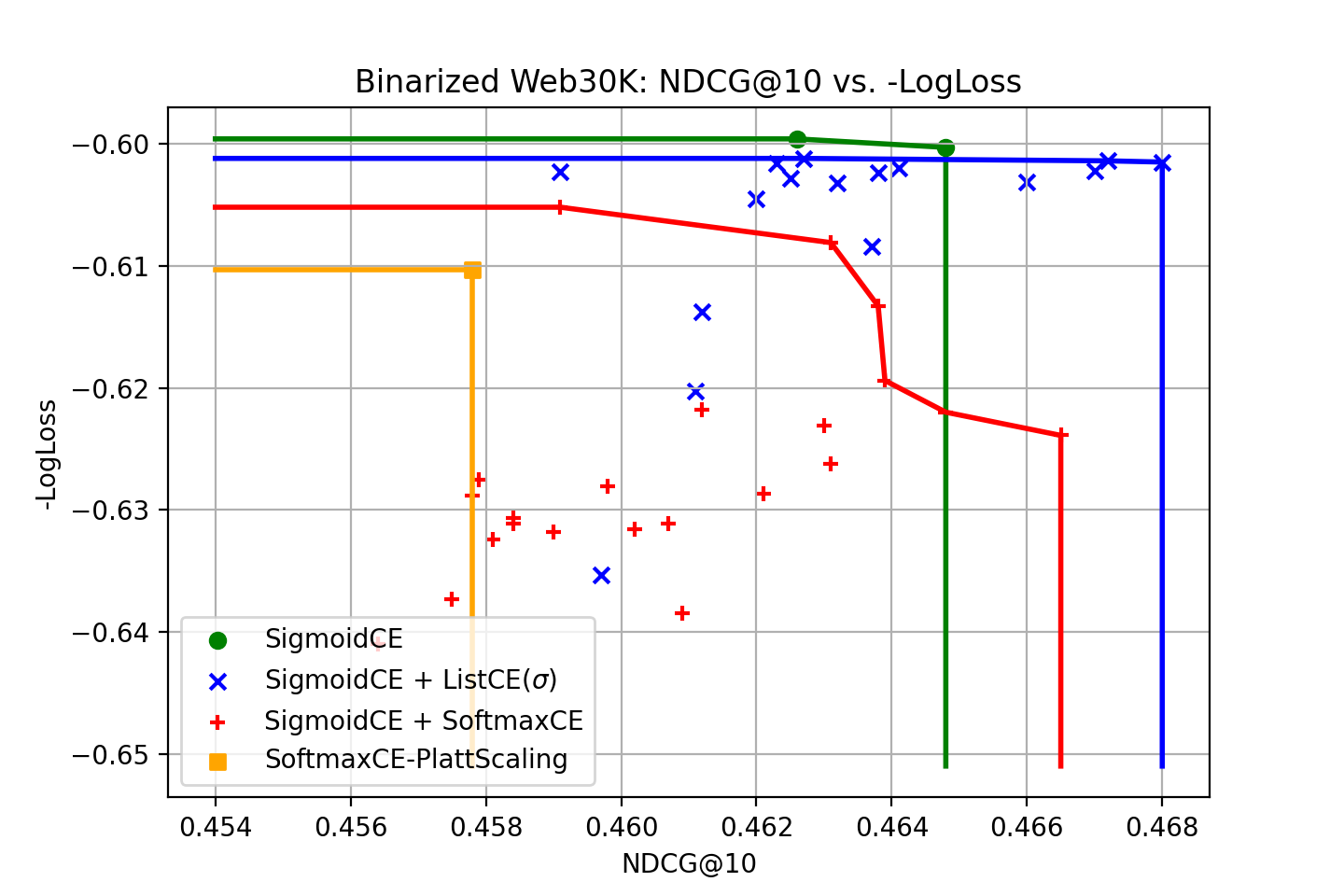}
         \caption{Binarized Web30K.}
         \label{fig:binarized_web30k}
     \end{subfigure}
     \begin{subfigure}[b]{0.33\textwidth}
         \centering
         \includegraphics[width=\textwidth]{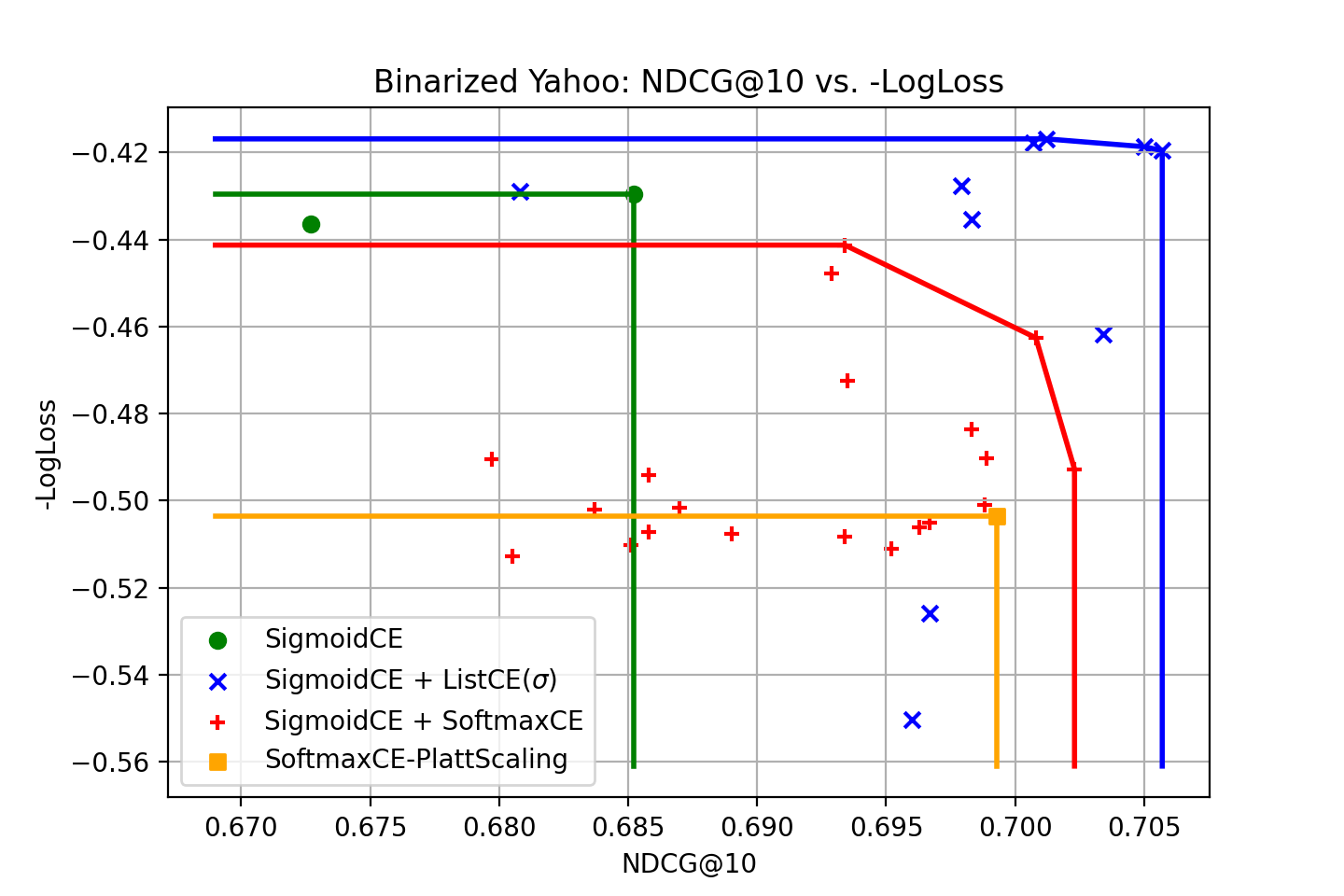}
         \caption{Binarized Yahoo.}
         \label{fig:binarized_yahoo}
     \end{subfigure}
     \begin{subfigure}[b]{0.33\textwidth}
         \centering
         \includegraphics[width=\textwidth]{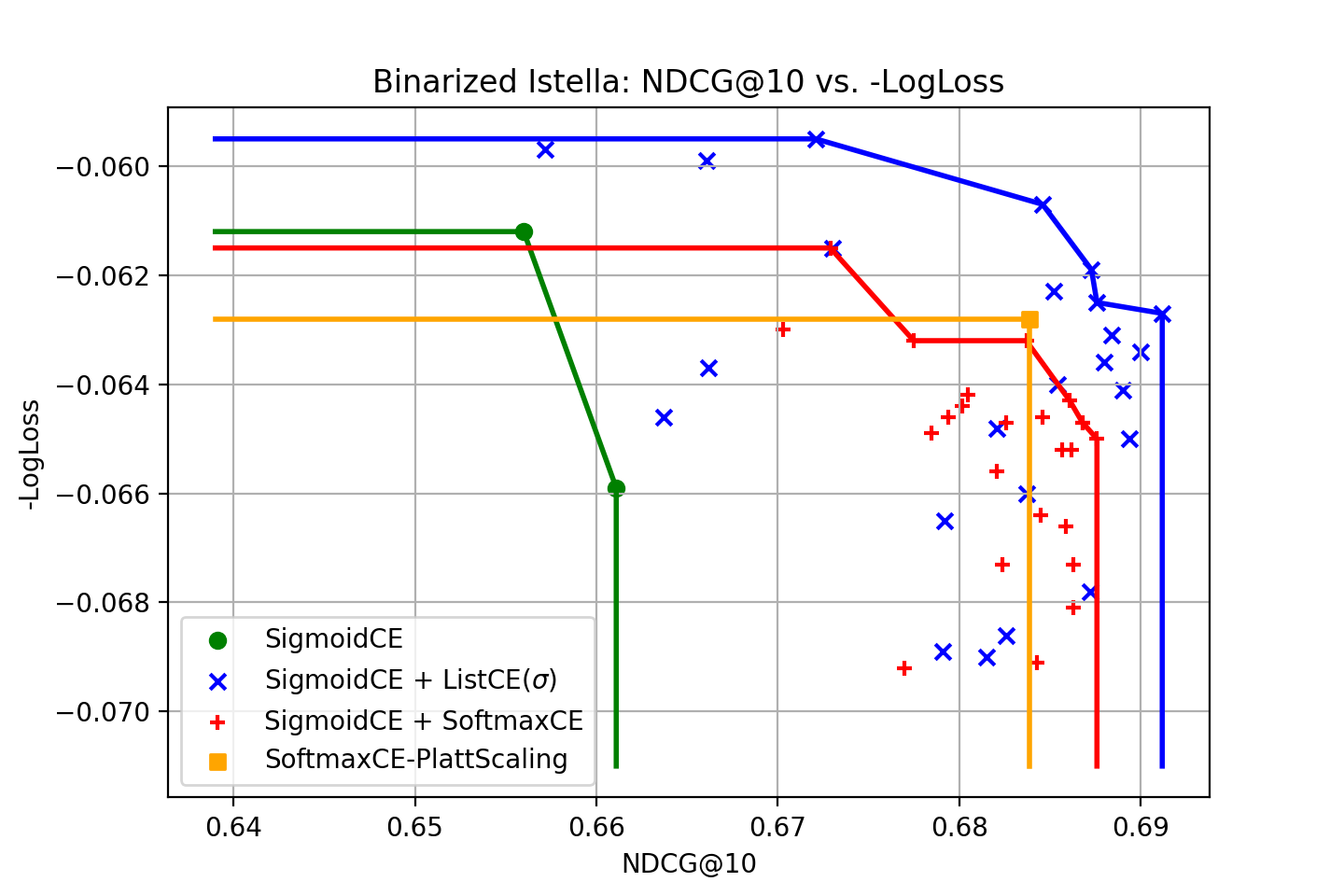}
         \caption{Binarized Istella.}
         \label{fig:binarized_istella}
     \end{subfigure}
        \caption{Pareto frontiers on the binarized Web30K, Yahoo and Istella datasets.}
        \label{fig:binarized}
\end{figure*}

The main results are shown in Table \ref{tab:binarized}. From the results, we can make the following observations:

\begin{itemize}
    \item The pointwise baselines, which are optimal in terms of logistic-regression losses by definition, consistently give either best or competitive regression performance in terms of both ECE and LogLoss; however, their ranking performance is usually inferior to other ranking-oriented methods, except in binarized Web30K where they archive better ranking performance than SoftmaxCE.
    \item The SoftmaxCE method can produce either best or competitive ranking performance; however, it is completely uncalibrated and performs poorly on regression metrics.
    \item The SoftmaxCE-Platt method performs well on regression metrics while giving the same ranking performance as SoftmaxCE. However, its regression and calibration performance is consistently inferior to the pointwise baselines.
    \item The standard multi-objective approach (SigmoidCE + SoftmaxCE) consistently achieves strong performance on both ranking and calibration metrics. It outperforms SoftmaxCE on all domains. This indicates that calibrated ranking scores can give better ranking performance than uncalibrated scores. In other words, the regression loss as a constraint in the multi-objective setting can help learning in ranking.
    \item The proposed RCR approach (SigmoidCE + ListCE($\sigma$)) consistently achieves the best ranking performance, while having comparable or better regression metrics than the pointwise baselines. This indicates that the compatible ranking and regression components within the RCR approach may mutually benefit each other and can achieve the top result on both fronts. It is also noticed that it consistently outperforms SoftmaxCE-Platt on regression metrics.
\end{itemize}

These observations indicate the proposed RCR approach is stable and performs well in terms of both regression and ranking metrics on a variety of configurations.

\subsubsection{Pareto Frontier Comparisons}
In the context of multi-objective optimization, the Pareto frontier is the set of Pareto optimal solutions where there is no scope for further Pareto improvement which is defined as a new solution where at least one objective gains, and no objectives lose. For each method, we evaluate all models over the hyper parameter space on the test data, plot each result as a regression-ranking metrics data point, and draw the Pareto frontier.

\begin{table*}[ht]
    \centering
    \caption{Comparisons with respect to relative differences in YouTube Search with SigmoidCE as the baseline.} \label{tab:youtube}
    \begin{tabular}{l||r|r||r|r|r|r}
    \hline
     & AUCPR & LogLoss &  NDCG@1 & NDCG@5 & NDCG@10 \\ \hline \hline
     Multi-Objective: $\SigmoidCE + \SoftmaxCE$ & -0.37\% & +0.13\% & \textbf{+0.30\%} & \textbf{+0.16\%} & \textbf{+0.15\%}  \\ \hline
     RCR (proposed): $\SigmoidCE + \ListCE(\sigma)$ & \textbf{+0.22\%} & \textbf{+0.03\%} & +0.27\% & +0.13\% & +0.13\%  \\ \hline
    \end{tabular}
\end{table*}

The results are shown in Figure \ref{fig:binarized}. Note that we use -LogLoss in the figures, so the Pareto frontier corresponds to the maxima of a point set. From the figures, we can see that RCR consistently dominates other methods in all domains except binarized Web30k. In binarized We30k, RCR dominates all other methods except SigmoidCE which gives better regression performance; however, its ranking performance is inferior to RCR. These results suggest that RCR can improve the Pareto frontiers or give new Pareto optimal solutions (e.g. better trade-offs) in a wide range of tasks.

\section{Experiments on YouTube Search} \label{sec:exp-yt}

We verify our approach on the real-world YouTube Search system, through both offline evaluations and large-scale online A/B testing.

\subsection{Background}
In YouTube Search, real-time user interaction data, represented as item-click pairs, is streaming to the model training infrastructure in a continuous way. Our baseline is a pCTR model that is equipped with the traditional SigmoidCE loss. Recently, new data with search page information, represented as page-item-click tuples, is made available to training, which gives the opportunity to directly improve its ranking quality within a search page. However, as stated previously, a direct switch from pointwise pCTR model to pairwise/listwise ranking model will not work in practice due to score calibration issues. It is required to improve the ranking quality of the model without affecting the values of its click predictions in any noticeable way.

\subsection{Experiments}
In this context, we compared the standard multi-objective approach ($\SigmoidCE + \SoftmaxCE$), and our RCR approach ($\SigmoidCE + \ListCE(\sigma)$) following the same setting in our baseline and train them continuously over the past $\sim$1 week of data over the same number of training steps for offline evaluation. The weight $\alpha$ is set to be 0.001 for both methods.

\paragraph{Offline Results} Offline results are reported in Table \ref{tab:youtube}. We use LogLoss and AUCPR to measure regression accuracy and NDCG for ranking quality, where AUCPR is defined as the area under the Precision-Recall curve. Higher AUCPR or lower LogLoss indicate better regression accuracy. Note that for proprietary reasons, we only report relative numbers to our baseline (SigmoidCE). From the results, we can see that the standard multi-objective approach improved the pCTR model on the NDCG ranking metric, but it caused significant degradation in both AUCPR and LogLoss metrics. Such a degradation can significantly affect the downstream stages negatively, making the models not suitable for the system. The proposed RCR approach not only improved the ranking quality, but also brought gains to pCTR predictions in terms of AUCPR. This is because in our approach the ranking component optimizes for the ranking capability directly in a way that is compatible with the regression component and is acting as a valid and aligned in-list constraint for regression -- which eventually helps the learning on regression. We also noticed that the proposed approach had a slight increase on LogLoss. This might be because the additional weight added on the ranking loss caused the learning on the regression loss to be less efficient than our baseline which is regression-only, thus the proposed approach may need more training steps for convergence.

\paragraph{Online A/B Testing} We further evaluated the model in a large-scale online A/B test over millions of users. The proposed model was tested against the production model. We report the following metrics in this experiment: {\bf SearchCTR(\%)} which is the percentage of clicks from search (the higher the better) and {\bf SearchAbandonedRate(\%)} which is the percentage of search queries that have 0 clicks and do not have refinements (the lower the better).

Due to data sensitivity, we only report relative performance of the experiment model over the production model. We observe that our proposed model improved SearchCTR by 0.66\%, and reduced SearchAbandonedRate by 0.31\% – which clearly indicates better search experiences for users. These metric gains are considered significant in the application. In comparison, the multi-objective approach doesn't qualify for production use because its negative impact on AUCPR and LogLoss. \textit{The proposed model has now been fully deployed to the YouTube Search system}.

The results from these experiments suggest that the proposed approach generalizes well to real-world production systems. The added ranking constraint not only improves ranking, but also benefits regression.

\section{Conclusion} \label{sec:conclusion}

In this paper, we propose the practical regression compatible ranking (RCR) approach for ranking tasks with binary relevance. Theoretically, we show that the regression and ranking components are mutually aligned in the sense that they share the same solution at global minima. Empirically, we show that RCR performs well on both regression and ranking metrics on several public LTR datasets, and significantly improves the Pareto frontiers in the context of multi-objective optimization. Furthermore, we show that RCR successfully improves both regression and ranking performance of a production pCTR model in YouTube Search and delivers better search experiences for users. We expect RCR to bring new opportunities for harmonious regression and ranking and to be applicable in a wide range of real-world applications where there is a list structure. In future work, we are interested in exploring more formulations for regression compatible ranking and beyond.

\bibliographystyle{ACM-Reference-Format}
\balance
\bibliography{references}


\begin{thebibliography}{33}


\ifx \showCODEN    \undefined \def \showCODEN     #1{\unskip}     \fi
\ifx \showDOI      \undefined \def \showDOI       #1{#1}\fi
\ifx \showISBNx    \undefined \def \showISBNx     #1{\unskip}     \fi
\ifx \showISBNxiii \undefined \def \showISBNxiii  #1{\unskip}     \fi
\ifx \showISSN     \undefined \def \showISSN      #1{\unskip}     \fi
\ifx \showLCCN     \undefined \def \showLCCN      #1{\unskip}     \fi
\ifx \shownote     \undefined \def \shownote      #1{#1}          \fi
\ifx \showarticletitle \undefined \def \showarticletitle #1{#1}   \fi
\ifx \showURL      \undefined \def \showURL       {\relax}        \fi
\providecommand\bibfield[2]{#2}
\providecommand\bibinfo[2]{#2}
\providecommand\natexlab[1]{#1}
\providecommand\showeprint[2][]{arXiv:#2}

\bibitem[Burges et~al\mbox{.}(2007)]%
        {lambdarank}
\bibfield{author}{\bibinfo{person}{Christopher Burges}, \bibinfo{person}{Robert
  Ragno}, {and} \bibinfo{person}{Quoc Le}.} \bibinfo{year}{2007}\natexlab{}.
\newblock \showarticletitle{Learning to Rank with Nonsmooth Cost Functions}. In
  \bibinfo{booktitle}{\emph{Advances in Neural Information Processing
  Systems}}. \bibinfo{publisher}{MIT Press}.
\newblock


\bibitem[Burges et~al\mbox{.}(2005)]%
        {ranknet2005icml}
\bibfield{author}{\bibinfo{person}{Chris Burges}, \bibinfo{person}{Tal Shaked},
  \bibinfo{person}{Erin Renshaw}, \bibinfo{person}{Ari Lazier},
  \bibinfo{person}{Matt Deeds}, \bibinfo{person}{Nicole Hamilton}, {and}
  \bibinfo{person}{Greg Hullender}.} \bibinfo{year}{2005}\natexlab{}.
\newblock \showarticletitle{Learning to Rank Using Gradient Descent}. In
  \bibinfo{booktitle}{\emph{Proceedings of the 22nd International Conference on
  Machine Learning}}. \bibinfo{pages}{89–96}.
\newblock


\bibitem[Burges(2010)]%
        {burges2010ranknet}
\bibfield{author}{\bibinfo{person}{Christopher~JC Burges}.}
  \bibinfo{year}{2010}\natexlab{}.
\newblock \showarticletitle{From RankNet to LambdaRank to LambdaMART: An
  Overview}.
\newblock \bibinfo{journal}{\emph{Learning}} \bibinfo{volume}{11},
  \bibinfo{number}{23-581} (\bibinfo{year}{2010}), \bibinfo{pages}{81}.
\newblock


\bibitem[Cao et~al\mbox{.}(2007)]%
        {cao2007learning}
\bibfield{author}{\bibinfo{person}{Zhe Cao}, \bibinfo{person}{Tao Qin},
  \bibinfo{person}{Tie-Yan Liu}, \bibinfo{person}{Ming-Feng Tsai}, {and}
  \bibinfo{person}{Hang Li}.} \bibinfo{year}{2007}\natexlab{}.
\newblock \showarticletitle{Learning to rank: from pairwise approach to
  listwise approach}. In \bibinfo{booktitle}{\emph{Proceedings of the 24th
  International Conference on Machine Learning}}. \bibinfo{pages}{129--136}.
\newblock


\bibitem[Chapelle and Chang(2011)]%
        {chapelle2011yahoo}
\bibfield{author}{\bibinfo{person}{Olivier Chapelle} {and} \bibinfo{person}{Yi
  Chang}.} \bibinfo{year}{2011}\natexlab{}.
\newblock \showarticletitle{Yahoo! learning to rank challenge overview}.
\newblock \bibinfo{journal}{\emph{Proceedings of Machine Learning Research}}
  \bibinfo{volume}{14} (\bibinfo{year}{2011}), \bibinfo{pages}{1--24}.
\newblock


\bibitem[Chaudhuri et~al\mbox{.}(2017)]%
        {Calibrate:ADKDD17}
\bibfield{author}{\bibinfo{person}{Sougata Chaudhuri}, \bibinfo{person}{Abraham
  Bagherjeiran}, {and} \bibinfo{person}{James Liu}.}
  \bibinfo{year}{2017}\natexlab{}.
\newblock \showarticletitle{Ranking and Calibrating Click-Attributed Purchases
  in Performance Display Advertising}. In \bibinfo{booktitle}{\emph{2017 AdKDD
  \& TargetAd}}. \bibinfo{pages}{7:1--7:6}.
\newblock


\bibitem[Dato et~al\mbox{.}(2016)]%
        {dato2016fast}
\bibfield{author}{\bibinfo{person}{Domenico Dato}, \bibinfo{person}{Claudio
  Lucchese}, \bibinfo{person}{Franco~Maria Nardini}, \bibinfo{person}{Salvatore
  Orlando}, \bibinfo{person}{Raffaele Perego}, \bibinfo{person}{Nicola
  Tonellotto}, {and} \bibinfo{person}{Rossano Venturini}.}
  \bibinfo{year}{2016}\natexlab{}.
\newblock \showarticletitle{Fast ranking with additive ensembles of oblivious
  and non-oblivious regression trees}.
\newblock \bibinfo{journal}{\emph{ACM Transactions on Information Systems
  (TOIS)}} \bibinfo{volume}{35}, \bibinfo{number}{2} (\bibinfo{year}{2016}),
  \bibinfo{pages}{1--31}.
\newblock


\bibitem[Guo et~al\mbox{.}(2017)]%
        {guo2017calibration}
\bibfield{author}{\bibinfo{person}{Chuan Guo}, \bibinfo{person}{Geoff Pleiss},
  \bibinfo{person}{Yu Sun}, {and} \bibinfo{person}{Kilian~Q Weinberger}.}
  \bibinfo{year}{2017}\natexlab{}.
\newblock \showarticletitle{On calibration of modern neural networks}. In
  \bibinfo{booktitle}{\emph{International conference on machine learning}}.
  PMLR, \bibinfo{pages}{1321--1330}.
\newblock


\bibitem[Jagerman et~al\mbox{.}(2019)]%
        {jagerman2019model}
\bibfield{author}{\bibinfo{person}{Rolf Jagerman}, \bibinfo{person}{Harrie
  Oosterhuis}, {and} \bibinfo{person}{Maarten de Rijke}.}
  \bibinfo{year}{2019}\natexlab{}.
\newblock \showarticletitle{To model or to intervene: A comparison of
  counterfactual and online learning to rank from user interactions}. In
  \bibinfo{booktitle}{\emph{Proceedings of the 42nd International ACM SIGIR
  Conference on Research and Development in Information Retrieval}}.
  \bibinfo{pages}{15--24}.
\newblock


\bibitem[Jagerman et~al\mbox{.}(2022)]%
        {jagerman2022optimizing}
\bibfield{author}{\bibinfo{person}{Rolf Jagerman}, \bibinfo{person}{Zhen Qin},
  \bibinfo{person}{Xuanhui Wang}, \bibinfo{person}{Mike Bendersky}, {and}
  \bibinfo{person}{Marc Najork}.} \bibinfo{year}{2022}\natexlab{}.
\newblock \showarticletitle{On Optimizing Top-K Metrics for Neural Ranking
  Models}. In \bibinfo{booktitle}{\emph{Proceedings of the 45th International
  ACM SIGIR Conference on Research and Development in Information Retrieval}}.
  \bibinfo{pages}{2303--2307}.
\newblock


\bibitem[J{\"a}rvelin and Kek{\"a}l{\"a}inen(2002)]%
        {jarvelin2002cumulated}
\bibfield{author}{\bibinfo{person}{Kalervo J{\"a}rvelin} {and}
  \bibinfo{person}{Jaana Kek{\"a}l{\"a}inen}.} \bibinfo{year}{2002}\natexlab{}.
\newblock \showarticletitle{Cumulated gain-based evaluation of IR techniques}.
\newblock \bibinfo{journal}{\emph{ACM Transactions on Information Systems
  (TOIS)}} \bibinfo{volume}{20}, \bibinfo{number}{4} (\bibinfo{year}{2002}),
  \bibinfo{pages}{422--446}.
\newblock


\bibitem[Joachims(2002)]%
        {Joachims:2002}
\bibfield{author}{\bibinfo{person}{Thorsten Joachims}.}
  \bibinfo{year}{2002}\natexlab{}.
\newblock \showarticletitle{Optimizing Search Engines Using Clickthrough Data}.
  In \bibinfo{booktitle}{\emph{Proceedings of the 8th ACM SIGKDD International
  Conference on Knowledge Discovery and Data Mining}}.
  \bibinfo{pages}{133--142}.
\newblock


\bibitem[Joachims et~al\mbox{.}(2017)]%
        {joachims2017unbiased}
\bibfield{author}{\bibinfo{person}{Thorsten Joachims}, \bibinfo{person}{Adith
  Swaminathan}, {and} \bibinfo{person}{Tobias Schnabel}.}
  \bibinfo{year}{2017}\natexlab{}.
\newblock \showarticletitle{Unbiased learning-to-rank with biased feedback}. In
  \bibinfo{booktitle}{\emph{Proceedings of the 10th ACM International
  Conference on Web Search and Data Mining}}. \bibinfo{pages}{781--789}.
\newblock


\bibitem[Ke et~al\mbox{.}(2017)]%
        {ke2017lightgbm}
\bibfield{author}{\bibinfo{person}{Guolin Ke}, \bibinfo{person}{Qi Meng},
  \bibinfo{person}{Thomas Finley}, \bibinfo{person}{Taifeng Wang},
  \bibinfo{person}{Wei Chen}, \bibinfo{person}{Weidong Ma},
  \bibinfo{person}{Qiwei Ye}, {and} \bibinfo{person}{Tie-Yan Liu}.}
  \bibinfo{year}{2017}\natexlab{}.
\newblock \showarticletitle{{LightGBM}: A highly efficient gradient boosting
  decision tree}. In \bibinfo{booktitle}{\emph{Advances in Neural Information
  Processing Systems}}. \bibinfo{pages}{3146--3154}.
\newblock


\bibitem[Kingma and Ba(2014)]%
        {kingma2014adam}
\bibfield{author}{\bibinfo{person}{Diederik~P Kingma} {and}
  \bibinfo{person}{Jimmy Ba}.} \bibinfo{year}{2014}\natexlab{}.
\newblock \showarticletitle{Adam: A method for stochastic optimization}.
\newblock \bibinfo{journal}{\emph{arXiv preprint arXiv:1412.6980}}
  (\bibinfo{year}{2014}).
\newblock


\bibitem[Li et~al\mbox{.}(2015)]%
        {ChengLi:KDD15}
\bibfield{author}{\bibinfo{person}{Cheng Li}, \bibinfo{person}{Yue Lu},
  \bibinfo{person}{Qiaozhu Mei}, \bibinfo{person}{Dong Wang}, {and}
  \bibinfo{person}{Sandeep Pandey}.} \bibinfo{year}{2015}\natexlab{}.
\newblock \showarticletitle{Click-through Prediction for Advertising in Twitter
  Timeline}. In \bibinfo{booktitle}{\emph{Proceedings of the 21th ACM SIGKDD
  International Conference on Knowledge Discovery and Data Mining}}.
  \bibinfo{pages}{1959–1968}.
\newblock


\bibitem[Liu et~al\mbox{.}(2009)]%
        {liu2009learning}
\bibfield{author}{\bibinfo{person}{Tie-Yan Liu} {et~al\mbox{.}}}
  \bibinfo{year}{2009}\natexlab{}.
\newblock \showarticletitle{Learning to rank for information retrieval}.
\newblock \bibinfo{journal}{\emph{Foundations and Trends{\textregistered} in
  Information Retrieval}} \bibinfo{volume}{3}, \bibinfo{number}{3}
  (\bibinfo{year}{2009}), \bibinfo{pages}{225--331}.
\newblock


\bibitem[Naeini et~al\mbox{.}(2015)]%
        {naeini2015obtaining}
\bibfield{author}{\bibinfo{person}{Mahdi~Pakdaman Naeini},
  \bibinfo{person}{Gregory Cooper}, {and} \bibinfo{person}{Milos Hauskrecht}.}
  \bibinfo{year}{2015}\natexlab{}.
\newblock \showarticletitle{Obtaining well calibrated probabilities using
  bayesian binning}. In \bibinfo{booktitle}{\emph{Proceedings of the 29th AAAI
  Conference on Artificial Intelligence}}, Vol.~\bibinfo{volume}{29}.
\newblock


\bibitem[Pang et~al\mbox{.}(2020)]%
        {pang2020setrank}
\bibfield{author}{\bibinfo{person}{Liang Pang}, \bibinfo{person}{Jun Xu},
  \bibinfo{person}{Qingyao Ai}, \bibinfo{person}{Yanyan Lan},
  \bibinfo{person}{Xueqi Cheng}, {and} \bibinfo{person}{Jirong Wen}.}
  \bibinfo{year}{2020}\natexlab{}.
\newblock \showarticletitle{Setrank: Learning a permutation-invariant ranking
  model for information retrieval}. In \bibinfo{booktitle}{\emph{Proceedings of
  the 43rd International ACM SIGIR Conference on Research and Development in
  Information Retrieval}}. \bibinfo{pages}{499--508}.
\newblock


\bibitem[Pasumarthi et~al\mbox{.}(2019)]%
        {pasumarthi2019tf}
\bibfield{author}{\bibinfo{person}{Rama~Kumar Pasumarthi},
  \bibinfo{person}{Sebastian Bruch}, \bibinfo{person}{Xuanhui Wang},
  \bibinfo{person}{Cheng Li}, \bibinfo{person}{Michael Bendersky},
  \bibinfo{person}{Marc Najork}, \bibinfo{person}{Jan Pfeifer},
  \bibinfo{person}{Nadav Golbandi}, \bibinfo{person}{Rohan Anil}, {and}
  \bibinfo{person}{Stephan Wolf}.} \bibinfo{year}{2019}\natexlab{}.
\newblock \showarticletitle{{TF-Ranking}: Scalable tensorflow library for
  learning-to-rank}. In \bibinfo{booktitle}{\emph{ACM SIGKDD International
  Conference on Knowledge Discovery and Data Mining}}.
  \bibinfo{pages}{2970--2978}.
\newblock


\bibitem[Platt(2000)]%
        {PlattScaling:1999}
\bibfield{author}{\bibinfo{person}{John Platt}.}
  \bibinfo{year}{2000}\natexlab{}.
\newblock \showarticletitle{Probabilistic Outputs for Support Vector Machines
  and Comparisons to Regularized Likelihood Methods}.
\newblock In \bibinfo{booktitle}{\emph{Advances in Large Margin Classifiers}},
  \bibfield{editor}{\bibinfo{person}{Alexander~J. Smola},
  \bibinfo{person}{Peter Bartlett}, \bibinfo{person}{Bernhard Schölkopf},
  {and} \bibinfo{person}{Dale Schuurmans}} (Eds.). \bibinfo{publisher}{MIT
  Press}, \bibinfo{pages}{61–74}.
\newblock


\bibitem[Qin and Liu(2013)]%
        {qin2013introducing}
\bibfield{author}{\bibinfo{person}{Tao Qin} {and} \bibinfo{person}{Tie-Yan
  Liu}.} \bibinfo{year}{2013}\natexlab{}.
\newblock \showarticletitle{Introducing LETOR 4.0 datasets}.
\newblock \bibinfo{journal}{\emph{arXiv preprint arXiv:1306.2597}}
  (\bibinfo{year}{2013}).
\newblock


\bibitem[Qin et~al\mbox{.}(2020)]%
        {matching}
\bibfield{author}{\bibinfo{person}{Zhen Qin}, \bibinfo{person}{Zhongliang Li},
  \bibinfo{person}{Michael Bendersky}, {and} \bibinfo{person}{Donald Metzler}.}
  \bibinfo{year}{2020}\natexlab{}.
\newblock \showarticletitle{Matching cross network for learning to rank in
  personal search}. In \bibinfo{booktitle}{\emph{Proceedings of The Web
  Conference 2020}}. \bibinfo{pages}{2835--2841}.
\newblock


\bibitem[Qin et~al\mbox{.}(2021)]%
        {dasalc}
\bibfield{author}{\bibinfo{person}{Zhen Qin}, \bibinfo{person}{Le Yan},
  \bibinfo{person}{Honglei Zhuang}, \bibinfo{person}{Yi Tay},
  \bibinfo{person}{Rama~Kumar Pasumarthi}, \bibinfo{person}{Xuanhui Wang},
  \bibinfo{person}{Michael Bendersky}, {and} \bibinfo{person}{Marc Najork}.}
  \bibinfo{year}{2021}\natexlab{}.
\newblock \showarticletitle{Are Neural Rankers still Outperformed by Gradient
  Boosted Decision Trees?}. In \bibinfo{booktitle}{\emph{Proceedings of the 9th
  International Conference on Learning Representations}}.
\newblock


\bibitem[Sculley(2010)]%
        {sculley2010combined}
\bibfield{author}{\bibinfo{person}{David Sculley}.}
  \bibinfo{year}{2010}\natexlab{}.
\newblock \showarticletitle{Combined regression and ranking}. In
  \bibinfo{booktitle}{\emph{Proceedings of the 16th ACM SIGKDD International
  Conference on Knowledge Discovery and Data Mining}}.
  \bibinfo{pages}{979--988}.
\newblock


\bibitem[Tagami et~al\mbox{.}(2013)]%
        {CTR:LTR:ADKDD13}
\bibfield{author}{\bibinfo{person}{Yukihiro Tagami}, \bibinfo{person}{Shingo
  Ono}, \bibinfo{person}{Koji Yamamoto}, \bibinfo{person}{Koji Tsukamoto},
  {and} \bibinfo{person}{Akira Tajima}.} \bibinfo{year}{2013}\natexlab{}.
\newblock \showarticletitle{CTR Prediction for Contextual Advertising:
  Learning-to-Rank Approach}. In \bibinfo{booktitle}{\emph{Proceedings of the
  7th International Workshop on Data Mining for Online Advertising}}. Article
  \bibinfo{articleno}{4}, \bibinfo{numpages}{8}~pages.
\newblock


\bibitem[Wang et~al\mbox{.}(2016)]%
        {wang2016learning}
\bibfield{author}{\bibinfo{person}{Xuanhui Wang}, \bibinfo{person}{Michael
  Bendersky}, \bibinfo{person}{Donald Metzler}, {and} \bibinfo{person}{Marc
  Najork}.} \bibinfo{year}{2016}\natexlab{}.
\newblock \showarticletitle{Learning to rank with selection bias in personal
  search}. In \bibinfo{booktitle}{\emph{Proceedings of the 39th International
  ACM SIGIR Conference on Research and Development in Information Retrieval}}.
  \bibinfo{pages}{115--124}.
\newblock


\bibitem[Wang et~al\mbox{.}(2018)]%
        {wang2018lambdaloss}
\bibfield{author}{\bibinfo{person}{Xuanhui Wang}, \bibinfo{person}{Cheng Li},
  \bibinfo{person}{Nadav Golbandi}, \bibinfo{person}{Michael Bendersky}, {and}
  \bibinfo{person}{Marc Najork}.} \bibinfo{year}{2018}\natexlab{}.
\newblock \showarticletitle{The {LambdaLoss} Framework for Ranking Metric
  Optimization}. In \bibinfo{booktitle}{\emph{Proceedings of the 27th ACM
  International Conference on Information and Knowledge Management}}.
  \bibinfo{pages}{1313--1322}.
\newblock


\bibitem[Xia et~al\mbox{.}(2008)]%
        {xia2008listwise}
\bibfield{author}{\bibinfo{person}{Fen Xia}, \bibinfo{person}{Tie-Yan Liu},
  \bibinfo{person}{Jue Wang}, \bibinfo{person}{Wensheng Zhang}, {and}
  \bibinfo{person}{Hang Li}.} \bibinfo{year}{2008}\natexlab{}.
\newblock \showarticletitle{Listwise approach to learning to rank: theory and
  algorithm}. In \bibinfo{booktitle}{\emph{Proceedings of the 25th
  International Conference on Machine Learning}}. \bibinfo{pages}{1192--1199}.
\newblock


\bibitem[Yan et~al\mbox{.}(2022a)]%
        {yan2022scale}
\bibfield{author}{\bibinfo{person}{Le Yan}, \bibinfo{person}{Zhen Qin},
  \bibinfo{person}{Xuanhui Wang}, \bibinfo{person}{Mike Bendersky}, {and}
  \bibinfo{person}{Marc Najork}.} \bibinfo{year}{2022}\natexlab{a}.
\newblock \showarticletitle{Scale Calibration of Deep Ranking Models}. In
  \bibinfo{booktitle}{\emph{Proceedings of the 28th ACM SIGKDD International
  Conference on Knowledge Discovery and Data Mining}}.
  \bibinfo{pages}{4300--4309}.
\newblock


\bibitem[Yan et~al\mbox{.}(2023)]%
        {yan2023learning}
\bibfield{author}{\bibinfo{person}{Le Yan}, \bibinfo{person}{Zhen Qin},
  \bibinfo{person}{Xuanhui Wang}, \bibinfo{person}{Gil Shamir}, {and}
  \bibinfo{person}{Mike Bendersky}.} \bibinfo{year}{2023}\natexlab{}.
\newblock \showarticletitle{Learning to Rank when Grades Matter}.
\newblock \bibinfo{journal}{\emph{arXiv preprint arXiv:2306.08650}}
  (\bibinfo{year}{2023}).
\newblock


\bibitem[Yan et~al\mbox{.}(2022b)]%
        {revisit}
\bibfield{author}{\bibinfo{person}{Le Yan}, \bibinfo{person}{Zhen Qin},
  \bibinfo{person}{Honglei Zhuang}, \bibinfo{person}{Xuanhui Wang},
  \bibinfo{person}{Michael Bendersky}, {and} \bibinfo{person}{Marc Najork}.}
  \bibinfo{year}{2022}\natexlab{b}.
\newblock \showarticletitle{Revisiting Two-Tower Models for Unbiased Learning
  to Rank}. In \bibinfo{booktitle}{\emph{Proceedings of the 45th International
  ACM SIGIR Conference on Research and Development in Information Retrieval}}.
  \bibinfo{pages}{2410–2414}.
\newblock


\bibitem[Zhu and Klabjan(2020)]%
        {zhu2020listwise}
\bibfield{author}{\bibinfo{person}{Xiaofeng Zhu} {and} \bibinfo{person}{Diego
  Klabjan}.} \bibinfo{year}{2020}\natexlab{}.
\newblock \showarticletitle{Listwise learning to rank by exploring unique
  ratings}. In \bibinfo{booktitle}{\emph{Proceedings of the 13th international
  conference on web search and data mining}}. \bibinfo{pages}{798--806}.
\newblock


\end{thebibliography}

\end{document}